\documentclass[11pt]{article}
\usepackage{hyperref}
\usepackage{times}  
\usepackage{mathpazo}
\usepackage{amssymb,amsmath,amsthm}
\usepackage{epsfig}
\usepackage{amsmath}
\usepackage{fullpage}

\newtheorem{theorem}{Theorem}[section]

\newtheorem{claim}[theorem]{Claim}
\newtheorem{lemma}[theorem]{Lemma}

\newtheorem{remark}[theorem]{Remark}

\newcommand{\qedsymb}{\hfill{\rule{2mm}{2mm}}}
\renewenvironment{proof}[1][]{\begin{trivlist}
\item[\hspace{\labelsep}{\bf\noindent Proof#1:\/}] }{\qedsymb\end{trivlist}}

\def\calB{{\cal B}}
\def\calC{{\cal C}}
\def\calG{{\cal G}}

\def\calZ{{\cal Z}}

\def\calC{{\cal C}}

\def\calE{{\cal E}}

\def\R{\mathbb{R}}

\newcommand\Alpha[3]{\alpha(#1,#2,#3)}

\newcommand\Prob[2]{{\Pr_{#1}\left[ {#2} \right]}}

\newcommand{\eps}{\epsilon}
\renewcommand{\epsilon}{\varepsilon}

\newcommand{\linspan}{\mathop{\mathrm{span}}}
\newcommand{\Fset}{\mathbb{F}}         

\begin{document}

\title{{\bf Nearly Orthogonal Sets over Finite Fields}
}
\author{
Dror Chawin\thanks{School of Computer Science, The Academic College of Tel Aviv-Yaffo, Tel Aviv 61083, Israel. Research supported by the Israel Science Foundation (grant No.~1218/20).}
\and
Ishay Haviv\footnotemark[1]
}

\date{}

\maketitle

\begin{abstract}
For a field $\Fset$ and integers $d$ and $k$, a set of vectors of $\Fset^d$ is called $k$-nearly orthogonal if its members are non-self-orthogonal and every $k+1$ of them include an orthogonal pair. We prove that for every prime $p$ there exists a positive constant $\delta = \delta (p)$, such that for every field $\Fset$ of characteristic $p$ and for all integers $k \geq 2$ and $d \geq k^{1/(p-1)}$, there exists a $k$-nearly orthogonal set of at least $d^{\delta \cdot k^{1/(p-1)}/ \log k}$ vectors of $\Fset^d$. In particular, for the binary field we obtain a set of $d^{\Omega( k /\log k)}$ vectors, and this is tight up to the $\log k$ term in the exponent. For comparison, the best known lower bound over the reals is $d^{\Omega( \log k / \log \log k)}$ (Alon and Szegedy,~Graphs and Combin.,~1999). The proof combines probabilistic and spectral arguments.
\end{abstract}

\section{Introduction}

For a field $\Fset$ and an integer $d$, two vectors $x,y \in \Fset^d$ are said to be orthogonal if they satisfy $\langle x, y \rangle = 0$ with respect to the inner product $\langle x, y \rangle = \sum_{i=1}^{d}{x_i y_i}$, where the arithmetic operations are over $\Fset$. If the vector $x$ satisfies $\langle x,x \rangle =0$, then we say that $x$ is self-orthogonal and otherwise, it is non-self-orthogonal.
For an integer $k$, a set $\calG \subseteq \Fset^d$ is called $k$-nearly orthogonal if its vectors are non-self-orthogonal and every subset of $k+1$ members of $\calG$ includes an orthogonal pair.
Let $\Alpha{d}{k}{\Fset}$ denote the maximum possible size of a $k$-nearly orthogonal subset of $\Fset^d$.
It can be easily seen that for $k=1$, it holds that $\Alpha{d}{1}{\Fset} = d$ for every field $\Fset$ and for every integer $d$. Indeed, the lower bound on $\Alpha{d}{1}{\Fset}$ follows by considering the $d$ vectors of the standard basis of $\Fset^d$, and the upper bound holds because a set of non-self-orthogonal vectors of $\Fset^d$ that are pairwise orthogonal is linearly independent, hence its size cannot exceed $d$.

A simple upper bound on $\Alpha{d}{k}{\Fset}$ follows from Ramsey theory.
Recall that for integers $s_1$ and $s_2$, the Ramsey number $R(s_1,s_2)$ is the smallest integer $r$ such that every graph on $r$ vertices has an independent set of size $s_1$ or a clique of size $s_2$.
For a given $k$-nearly orthogonal set $\calG \subseteq \Fset^d$, consider the graph on the vertex set $\calG$, where two distinct vertices are adjacent if and only if their vectors are not orthogonal.
Since the vectors of $\calG$ are non-self-orthogonal and lie in $\Fset^d$, this graph has no independent set of size $d+1$, and since every $k+1$ members of $\calG$ include an orthogonal pair, this graph has no clique of size $k+1$. We thus obtain that
\begin{eqnarray}\label{eq:Ramsey}
\Alpha{d}{k}{\Fset} < R(d+1,k+1) \leq \binom{d+k}{k},
\end{eqnarray}
where the second inequality follows by a famous upper bound on Ramsey numbers due to Erd{\H{o}}s and Szekeres~\cite{ErdosS35}.
Note that for every fixed constant $k$, it follows from~\eqref{eq:Ramsey} that $\Alpha{d}{k}{\Fset} \leq O(d^k)$.

The problem of determining the values of $\Alpha{d}{k}{\Fset}$ for $\Fset$ being the real field $\R$ was proposed by Erd{\H{o}}s in 1988 (see~\cite{NR97Erdos}).
For all integers $d$ and $k$, it holds that
\begin{eqnarray}\label{eq:alpha>=kd}
\Alpha{d}{k}{\R} \geq k \cdot d,
\end{eqnarray}
as follows by considering the union of $k$ pairwise disjoint orthogonal bases of $\R^d$.
Erd{\H{o}}s conjectured that this bound is tight for $k=2$, that is, $\Alpha{d}{2}{\R}=2 \cdot d$, and his conjecture was confirmed a few years later by Rosenfeld~\cite{Rosenfeld91} (see also~\cite{Deaett11}). In 1992, F{\"{u}}redi and Stanley~\cite{FurediS92} showed that $\Alpha{4}{5}{\R} \geq 24$, which implies that the bound in~\eqref{eq:alpha>=kd} is not tight in general, and conjectured that there exists a constant $c$ such that $\Alpha{d}{k}{\R} \leq (k \cdot d)^c$ for all integers $d$ and $k$. Their conjecture was disproved in 1999 by Alon and Szegedy~\cite{AlonS99}, who applied the probabilistic method to prove that there exists a constant $\delta > 0$ such that for all integers $d$ and $k \geq 3$, it holds that
\[\Alpha{d}{k}{\R} \geq d^{\delta \cdot \log k/ \log \log k},\]
where here and throughout, all logarithms are in base $2$.
More recently, Balla, Letzter, and Sudakov~\cite{BallaLS20} proved that $\Alpha{d}{k}{\R} \leq O(d^{(k+1)/3})$ for every fixed constant $k$, improving on the upper bound of~\eqref{eq:Ramsey} for the real field. Balla~\cite{Balla23} further provided a bipartite analogue of the aforementioned result of~\cite{AlonS99}.
Nevertheless, the currently known upper and lower bounds on $\Alpha{d}{k}{\R}$ for general $d$ and $k$ are rather far apart (see, e.g.,~\cite[Chapter~10.2]{LovaszBook}).

Soon after the appearance of~\cite{AlonS99}, Codenotti, Pudl{\'{a}}k, and Resta~\cite{CodenottiPR00} considered the question of determining the values of $\Alpha{d}{k}{\Fset}$ for finite fields $\Fset$, focusing on the case where $\Fset$ is the binary field $\Fset_2$ and $k=2$. Note that $\Alpha{d}{2}{{\Fset_2}}$ can be formulated as the largest possible size of a family of subsets of $[d]=\{1,\ldots,d\}$, where each subset has an odd size and among every three of the subsets there are two that intersect at an even number of elements.
While the original motivation of~\cite{CodenottiPR00} to construct large $2$-nearly orthogonal sets over $\Fset_2$ arrived from the rigidity approach to lower bounds in circuit complexity, this challenge enjoys additional diverse applications from the area of information theory, related to distributed storage, index coding, and hat-guessing games (see, e.g.,~\cite{BlasiakKL13,BargZ22,BargSY22,HuangX23} and Section~\ref{sec:related}).

Codenotti et al.~\cite{CodenottiPR00} provided an explicit construction of a $2$-nearly orthogonal subset of $\Fset_2^d$ of size $4d-8$ for every $d \geq 2$, implying that $\Alpha{d}{2}{{\Fset_2}} \geq 4d-8$, and asked whether this bound is tight. Another explicit construction, of size $16 \cdot \lfloor d/6 \rfloor$, was provided by Blasiak, Kleinberg, and Lubetzky~\cite{BlasiakKL13}.
The question of~\cite{CodenottiPR00} was answered negatively in~\cite{GolovnevH20}, where it was shown by an explicit construction that $\Alpha{d}{2}{{\Fset_2}} \geq d^{1+\delta}$ for some constant $\delta >0$ and for infinitely many integers $d$.
Recalling that $\Alpha{d}{2}{\R}=2 \cdot d$, this demonstrates that the value of $\Alpha{d}{k}{\Fset}$ behaves quite differently for different fields $\Fset$ already for $k=2$.

\subsection{Our Contribution}

In the present paper, we prove lower bounds on the quantities $\Alpha{d}{k}{\Fset}$ for finite fields $\Fset$.
For the binary field $\Fset_2$, we prove the following result.

\begin{theorem}\label{thm:F2}
There exists a constant $\delta >0$ such that for all integers $d \geq k \geq 2$, it holds that
\[ \Alpha{d}{k}{{\Fset_2}} \geq d^{\delta \cdot k/\log k}.\]
\end{theorem}
\noindent
The bound provided by Theorem~\ref{thm:F2} gets close to the upper bound given in~\eqref{eq:Ramsey} and is tight up to the $\log k$ term in the exponent.
Note that the assumption that $d$ is sufficiently large compared to $k$ is essential, because otherwise the bound guaranteed by the theorem might exceed the total number of vectors in $\Fset_2^d$.
It is interesting to compare Theorem~\ref{thm:F2} to the best known lower bound on $\Alpha{d}{k}{\Fset}$ where $\Fset$ is the real field $\R$, namely, the bound $d^{\Omega( \log k/ \log \log k)}$ achieved in~\cite{AlonS99}.

We proceed with an extension of Theorem~\ref{thm:F2} to general finite fields (in fact, fields with finite characteristic), stated as follows.

\begin{theorem}\label{thm:GeneralFields}
For every prime $p$ there exists a constant $\delta = \delta(p) >0$, such that for every field $\Fset$ of characteristic $p$ and for all integers $k \geq 2$ and $d \geq k^{1/(p-1)}$, it holds that
\[ \Alpha{d}{k}{\Fset} \geq d^{\delta \cdot k^{1/(p-1)}/\log k}.\]
\end{theorem}
\noindent
Note that the special case of Theorem~\ref{thm:GeneralFields} with $p=2$ coincides with Theorem~\ref{thm:F2}.

In fact, inspired by a recent paper of Balla~\cite{Balla23}, we provide a bipartite analogue of Theorem~\ref{thm:GeneralFields} with the same asymptotic bound.
The bipartite setting concerns sets $\calG$ of non-self-orthogonal vectors of $\Fset^d$, such that for every two subsets $G_1,G_2 \subseteq \calG$ of size $k$ each, some vector of $G_1$ is orthogonal to some vector of $G_2$ (see Theorem~\ref{thm:GeneralFieldsBipartite}). Notice that such a set $\calG$ is $(k-1)$-nearly orthogonal, hence our result for the bipartite setting strengthens Theorem~\ref{thm:GeneralFields}.
Our bipartite analogue of Theorem~\ref{thm:GeneralFields} is motivated by questions on the dimension of orthogonal representations over finite fields of $H$-free graphs, where $H$ is the complete bipartite graph $K_{k,k}$.
We describe these questions and related work in Section~\ref{sec:related}. The applications of our results to this context are presented in Section~\ref{sec:OD}.
Among other things, we determine up to a multiplicative constant the largest possible ratio between the clique cover number and the minimum dimension of an orthogonal representation over the binary field for graphs on $n$ vertices (see Theorem~\ref{thm:CHIvsOD} and Remark~\ref{remark:CHI}).

The proofs of Theorems~\ref{thm:F2} and~\ref{thm:GeneralFields} borrow the probabilistic approach of Alon and Szegedy~\cite{AlonS99} for producing $k$-nearly orthogonal sets over $\R$, which relies on the randomized graph product technique developed by Berman and Schnitger~\cite{BermanS92} and by Feige~\cite{Feige97}. However, in order to use this approach in the finite field setting and to establish significantly larger $k$-nearly orthogonal sets, we combine it with several additional ideas. For Theorem~\ref{thm:F2}, we combine the approach of~\cite{AlonS99} with elementary linear-algebraic arguments.
For Theorem~\ref{thm:GeneralFields} and for its bipartite analogue, we further use spectral techniques to prove pseudo-random properties of a related graph family. The analysis involves a result of Le Anh Vinh~\cite{Vinh08a}, which builds on a result of Alon and Krivelevich~\cite{AlonK97}.
As a by-product of our proof technique, we obtain an essentially tight estimation for the number of sets of non-self-orthogonal vectors that are pairwise non-orthogonal in vector spaces over prime order fields (see Theorem~\ref{thm:counting}).

\subsection{Proof Overview}

We offer here a high-level description of the proofs of Theorems~\ref{thm:F2} and~\ref{thm:GeneralFields}.
We start by presenting the approach of Alon and Szegedy~\cite{AlonS99} in their proof for the existence of large $k$-nearly orthogonal sets over $\R$, and then explain how we adapt it to the finite field setting and how we establish significantly larger sets.

For some integers $t$ and $m$, let $V$ denote the set of all $2^t$ real vectors of length $t$ whose values are in $\{-1,+1\}$, and let $Q$ denote the set of all vectors of the form $v_1 \otimes \cdots \otimes v_m$ for vectors $v_1, \ldots, v_m$ of $V$, where $\otimes$ stands for the tensor product operation on vectors (for the definition and properties of this operation, see Section~\ref{sec:preliminaries}). Note that the vectors of $Q$ are of length $t^m$. For some integer $n$, let $\calG$ be a random set of $n$ vectors chosen uniformly and independently from $Q$.
The argument proceeds by estimating the probability that $\calG$ is $k$-nearly orthogonal.

The probabilistic analysis in~\cite{AlonS99} uses a result of Frankl and R{\"o}dl~\cite{FranklR87}, which asserts that for some $\epsilon >0$, every subset of $V$ whose vectors are pairwise non-orthogonal has size at most $2^{(1-\eps) \cdot t}$.
A key observation, which follows by standard properties of the tensor product operation, is that every subset of $Q$ whose vectors are pairwise non-orthogonal is contained in a set of the form
\begin{eqnarray}\label{eq:box}
\{v_1 \otimes \cdots \otimes v_m \mid v_j \in A_j \mbox{~for all~} j \in [m]\}
\end{eqnarray}
for sets $A_1, \ldots, A_m \subseteq V$, where the vectors of each $A_j$ are pairwise non-orthogonal.
Referring to a set of the form~\eqref{eq:box} as a box, it suffices to show that with positive probability, no box contains $k+1$ of the random members of $\calG$.
The above result of~\cite{FranklR87} implies that the size of every box is at most $2^{m \cdot (1-\eps) \cdot t}$, hence the probability that a random vector chosen from $Q$ lies in a fixed box does not exceed $2^{m \cdot (1-\eps) \cdot t}/|Q| \approx 2^{-\eps \cdot m \cdot t}$. By applying the union bound, first over all choices of $k+1$ members of $\calG$ for a fixed box, and then over all boxes, one may obtain an upper bound on the probability that $\calG$ includes $k+1$ vectors that lie in a common box and thus on the probability that $\calG$ is not $k$-nearly orthogonal.
Note that the number of boxes does not exceed the number of $m$-tuples of subsets of $V$, hence it is at most $2^{m \cdot 2^t}$.
It turns out that for every integer $k \geq 3$, it is possible to choose the values of $t$ and $m$ such that for $d=t^m$ and $n = d^{\Omega(\log k / \log \log k)}$, the set $\calG \subseteq \R^d$ is $k$-nearly orthogonal with positive probability. This yields the existence of the desired $k$-nearly orthogonal sets over $\R$ and completes the description of the argument of~\cite{AlonS99}.

We next describe our proof of Theorem~\ref{thm:F2}, which asserts the existence of large $k$-nearly orthogonal sets over the binary field $\Fset_2$.
Following the approach of~\cite{AlonS99}, for integers $t$ and $m$, we consider the set $V$ of all non-self-orthogonal vectors in $\Fset_2^t$ and the set $Q$ of all the vectors in $\Fset_2^{t^m}$ that can be represented as a tensor product of $m$ vectors of $V$. Note that the vectors of $Q$ are non-self-orthogonal. As before, we let $\calG$ be a random set of $n$ vectors chosen uniformly and independently from $Q$, and we estimate the probability that $\calG$ is $k$-nearly orthogonal.

The first question to ask here is how large can be a subset of $V$ whose vectors are pairwise non-orthogonal.
A simple linear-algebraic argument shows that its size cannot exceed $2^{(t+1)/2}$. Indeed, consider a set of pairwise non-orthogonal vectors of $V$, and add a $1$ entry at the end of each of them. This gives us a set of self-orthogonal vectors of $\Fset_2^{t+1}$ that are pairwise orthogonal. It follows that the subspace of $\Fset_2^{t+1}$ spanned by those vectors is contained in its orthogonal complement, and as such, its dimension cannot exceed $(t+1)/2$. This yields the desired bound of $2^{(t+1)/2}$ on the size of the given set. In fact, this argument not only provides a bound on the size of any set of pairwise non-orthogonal vectors of $V$, but also provides valuable information about its structure. Namely, every such set is contained in some subspace of $\Fset_2^t$ of dimension at most $(t+1)/2$, obtained by omitting the last coordinate of the subspace of $\Fset_2^{t+1}$ considered above (see Lemma~\ref{lemma:collectionC_F2}).

As in the case of the real field, one may observe that every subset of $Q$ whose vectors are pairwise non-orthogonal is contained in a box defined as in~\eqref{eq:box}.
Moreover, the size of every such box is bounded by $2^{m \cdot (t+1)/2}$, hence the probability that a random vector chosen from $Q$ lies in a fixed box is bounded by $2^{m \cdot (t+1)/2}/|Q| \approx 2^{-m \cdot t/2}$.
By applying the union bound, first over all choices of $k+1$ members of $\calG$ for a fixed box, and then over all boxes, it is possible to bound the probability that some $k+1$ vectors of $\calG$ lie in a common box and thus to bound the probability that $\calG$ is not $k$-nearly orthogonal.
The crucial point that allows us to obtain significantly larger sets over $\Fset_2$ than over the reals is a stronger bound on the number of needed boxes.
As explained above, every set of pairwise non-orthogonal vectors of $V$ is contained in a subspace of $\Fset_2^t$ of dimension bounded away from $t$. The number of such subspaces clearly does not exceed $2^{t^2}$, hence the number of boxes that the union bound should take care of is at most $2^{m \cdot t^2}$ (in comparison to the $2^{m \cdot 2^t}$ boxes over $\R$). This allows us to show that for every integer $k \geq 2$, it is possible to choose the values of $t$ and $m$ such that for $d=t^m$ and $n = d^{\Omega(k / \log k)}$, the set $\calG \subseteq \Fset_2^d$ is $k$-nearly orthogonal with positive probability.

We next discuss the proof of Theorem~\ref{thm:GeneralFields} for finite fields of prime order.
First, note that the argument described above does not extend in a straightforward manner to the field $\Fset_p$ for a prime $p>2$.
Specifically, it no longer holds that every set of non-self-orthogonal vectors of $\Fset_p^t$ that are pairwise non-orthogonal is contained in a subspace of dimension bounded away from $t$. To see this, consider the set $\{e_1, e_1+e_2, e_1+e_3,\ldots, e_1+e_t\}$, where $e_i$ stands for the $i$th vector of the standard basis of $\Fset_p^t$.
For every prime $p > 2$, the vectors of this set are non-self-orthogonal and are pairwise non-orthogonal, and yet, they span the entire space $\Fset_p^t$.
Nevertheless, we show that similarly to the binary case, the size of a set of non-self-orthogonal vectors of $\Fset_p^t$ that are pairwise non-orthogonal is bounded from above by roughly $p^{t/2}$ (see Theorem~\ref{thm:vectors_B,C}). This result is established using a spectral argument. Namely, we consider a suitable family of graphs and derive their pseudo-random properties using the second largest eigenvalue determined in~\cite{Vinh08a} (see Section~\ref{sec:spectral}).

However, in order to apply the probabilistic argument with a relatively small collection of boxes, it is not sufficient to bound the size of sets of non-self-orthogonal vectors of $\Fset_p^t$ that are pairwise non-orthogonal. We further need some understanding of the structure of those sets. Specifically, one has to come up with a relatively small collection of relatively small subsets of $\Fset_p^t$, such that every set of non-self-orthogonal vectors of $\Fset_p^t$ that are pairwise non-orthogonal is contained in one of them.
To do so, we show that for some integer $s = O(t^{p-1})$, there exists a function $g: \Fset_p^t \rightarrow \Fset_p^s$ satisfying that for all $v_1,v_2 \in \Fset_p^t$, the vectors $v_1$ and $v_2$ are orthogonal if and only if their images $g(v_1)$ and $g(v_2)$ are not. We use this function to produce the required collection, and it turns out that its size is bounded by the number of subspaces of $\Fset_p^s$ that are spanned by vectors of the image of $g$, and thus by $(p^{t})^s = p^{O(t^{p})}$ (see Lemma~\ref{lemma:collectionC}).
Then, the number of boxes on which the union bound has to be applied is $p^{O(m \cdot t^{p})}$, resulting in the $k$-nearly orthogonal sets declared in Theorem~\ref{thm:GeneralFields}.

\subsection{Related Work}\label{sec:related}

This paper is concerned with the problem of determining the largest possible size $\Alpha{d}{k}{\Fset}$ of a $k$-nearly orthogonal subset of $\Fset^d$ for a field $\Fset$ and integers $d$ and $k$. In addition to the requirement that every $k+1$ members of such a set include an orthogonal pair, we require the vectors of the set to be non-self-orthogonal. Over the reals, the latter condition simply means that the vectors are nonzero, but over finite fields, it significantly affects the problem. Without this requirement, the upper bound~\eqref{eq:Ramsey} which relies on Ramsey theory does not hold, and the size of the sets may grow exponentially in $d$ even for $k=1$. The maximum size of a set of pairwise orthogonal vectors, which are not necessarily non-self-orthogonal, was determined precisely in the late sixties for the binary field by Berlekamp~\cite{Berlekamp69} and for any prime order field by Zame~\cite{Zame70}. Further extensions to general bilinear forms were recently provided by Mohammadi and Petridis~\cite{MohammadiP22}.

Another variant of the problem asks, for a given integer $\ell$, to find large sets of non-self-orthogonal vectors of $\Fset^d$ such that every $k+1$ of them include $\ell+1$ pairwise orthogonal vectors. While the present paper focuses on the case $\ell = 1$, the work of Alon and Szegedy~\cite{AlonS99} does consider this general setting over the reals and shows that for every integer $\ell \geq 1$ there exists a constant $\delta = \delta(\ell) > 0$, for which there exists a set of at least $d^{\delta \cdot \log k / \log \log k}$ vectors of $\R^d$ satisfying the above property. The challenge of extending our results to this setting is left for future research.

The quantities $\Alpha{d}{k}{\Fset}$ can be represented in terms of orthogonal representations of graphs, a notion that was proposed in the seminal paper of Lov{\'{a}}sz~\cite{Lovasz79} that introduced the celebrated $\vartheta$-function. A $d$-dimensional orthogonal representation of a graph over a field $\Fset$ is an assignment of a non-self-orthogonal vector of $\Fset^d$ to each vertex, such that the vectors assigned to two distinct non-adjacent vertices are orthogonal. As mentioned earlier, one may associate with every set $\calG$ of non-self-orthogonal vectors of $\Fset^d$, a graph on the vertex set $\calG$ with edges connecting pairs of non-orthogonal vectors.
Notice that such a $\calG$ is $k$-nearly orthogonal if and only if the graph associated with it is $K_{k+1}$-free, i.e., contains no copy of the complete graph on $k+1$ vertices. Therefore, $\Alpha{d}{k}{\Fset}$ is closely related to the largest possible number of vertices in a $K_{k+1}$-free graph that admits a $d$-dimensional orthogonal representation over $\Fset$ (with a minor difference between the quantities, caused by the fact that the vectors of an orthogonal representation are not necessarily distinct). The problem of determining the smallest $d=d(n)$ such that every $H$-free graph on $n$ vertices has a $d$-dimensional orthogonal representation over the reals was studied in the literature for various graphs $H$ with a particular attention to cycles (see, e.g.,~\cite{Pudlak02,BallaLS20}). Note that the bipartite setting studied in~\cite{Balla23} corresponds to the case where $H$ is a balanced complete bipartite graph. The applications of the results of the present paper to this context are given in Section~\ref{sec:OD}.

It is worth mentioning here the similar concept of $d$-dimensional orthogonal bi-representations of graphs over a field $\Fset$.
Here, one has to assign to each vertex a pair of non-orthogonal vectors of $\Fset^d$, such that the pairs $(u_1,u_2)$ and $(v_1,v_2)$ assigned to distinct non-adjacent vertices satisfy $\langle u_1, v_2 \rangle = \langle v_1, u_2 \rangle = 0$. Note that $d$-dimensional orthogonal representations can be viewed as a special case of $d$-dimensional orthogonal bi-representations by replacing every vector $v$ with the pair $(v,v)$. This notion was proposed by Peeters~\cite{Peeters96} to provide an alternative definition for the minrank parameter of graphs, introduced by Haemers~\cite{Haemers78} in the study of the Shannon capacity. While we do not mention here the original definition of the minrank of a graph over a field $\Fset$, it turns out that the latter is precisely the smallest integer $d$ for which the graph admits a $d$-dimensional orthogonal bi-representation over $\Fset$. This quantity of graphs has attracted considerable attention due to its various applications in computational complexity and in information theory. Notably, the question of determining the smallest possible minrank over a given field $\Fset$ over all $H$-free graphs on $n$ vertices has been explored for various graphs $H$ (see, e.g.,~\cite{CodenottiPR00,AlonHKPY11,Haviv18free,GolovnevH20}). The particular setting in which $\Fset$ is the binary field and $H$ is the triangle graph is the focus of an intensive recent line of work (see~\cite{BargZ22,BargSY22,HuangX23}).

\subsection{Outline}
The rest of the paper is organized as follows.
In Section~\ref{sec:preliminaries}, we provide some useful properties of the tensor product operation on vectors.
In Section~\ref{sec:F2}, we prove Theorem~\ref{thm:F2}. Its proof is elementary, and we present it separately from the proof of Theorem~\ref{thm:GeneralFields} for didactic reasons.
Then, in Section~\ref{sec:F}, we give some background on spectral graph theory and use it to extend Theorem~\ref{thm:F2} to general finite fields and to the bipartite setting, and in particular, to establish Theorem~\ref{thm:GeneralFields}.
We also present there applications of our results to the context of orthogonal representations of graphs.

\section{Preliminaries}\label{sec:preliminaries}

Our proofs crucially use the tensor product operation on vectors.
For a field $\Fset$ and integers $t_1,t_2$, the tensor product $w = u \otimes v$ of two vectors $u \in \Fset^{t_1}$ and $v \in \Fset^{t_2}$ is defined as the vector in $\Fset^{t_1 \cdot t_2}$, whose coordinates are indexed by the pairs $(i_1,i_2)$ with $i_1 \in [t_1]$ and $i_2 \in [t_2]$ (ordered lexicographically), defined by $w_{(i_1,i_2)} = u_{i_1} \cdot v_{i_2}$. Note that the representation of $w$ as a tensor product of two vectors of lengths $t_1$ and $t_2$ is not necessarily unique. Note further that for integers $t$ and $m$ and for given vectors $v_1, \ldots, v_m \in \Fset^t$, the vector $v_1 \otimes \cdots \otimes v_m$ lies in $\Fset^{t^m}$ and consists of all the $t^m$ possible products of $m$ values, one taken from each vector $v_j$ with $j \in [m]$. When all the vectors $v_1, \ldots, v_m$ are equal to a single vector $v$, their tensor product $v_1 \otimes \cdots \otimes v_m$ can be written as $v^{\otimes m}$.

It is well known and easy to verify that for vectors $u_1, \ldots, u_m \in \Fset^t$
and $v_1, \ldots, v_m \in \Fset^t$, the two vectors $u = u_1 \otimes \cdots \otimes u_m$ and $v = v_1 \otimes \cdots \otimes v_m$ satisfy
\begin{eqnarray}\label{eq:tensor}
\langle u, v \rangle = \prod_{j=1}^{m}{\langle u_j , v_j \rangle}.
\end{eqnarray}
It thus follows that $u$ and $v$ are orthogonal if and only if $u_j$ and $v_j$ are orthogonal for some $j \in [m]$. In particular, the vector $u$ is self-orthogonal if and only if $u_j$ is self-orthogonal for some $j \in [m]$.

For $m$ sets $A_1, \ldots, A_m \subseteq \Fset^t$, we let $A_1 \otimes \cdots \otimes A_m$ denote the collection of all vectors of the form $v_1 \otimes \cdots \otimes v_m$ with $v_j \in A_j$ for each $j \in [m]$.
While this notation will be convenient for us throughout this paper, let us stress that it does not coincide with the usual tensor product operation on linear subspaces.
As before, when all the sets $A_1, \ldots, A_m$ are equal to a single set $A$, the collection $A_1 \otimes \cdots \otimes A_m$ can be written as $A^{\otimes m}$.
For a set $G \subseteq \Fset^{t^m}$ whose vectors are represented as tensor products of $m$ vectors of $\Fset^t$, the $j$th projection of $G$ is the set of all vectors $v \in \Fset^t$ for which there exists a member $v_1 \otimes \cdots \otimes v_m$ of $G$ with $v_j = v$.

We will need the following simple claim.
\begin{claim}\label{claim:|V|^m}
For a field $\Fset$ and two integers $t$ and $m$, let $A_1, \ldots, A_m$ be sets of nonzero vectors of $\Fset^t$ such that the first nonzero value in each vector is $1$.
Then, $|A_1 \otimes \cdots \otimes A_m| = \prod_{j=1}^{m}{|A_j|}$.
\end{claim}
\begin{proof}
Fix a field $\Fset$.
We will prove that for all integers $t_1$ and $t_2$ and for all sets $A_1 \subseteq \Fset^{t_1}$ and $A_2 \subseteq \Fset^{t_2}$ of nonzero vectors whose first nonzero value is $1$, it holds that
\begin{enumerate}
  \item the vectors of $A_1 \otimes A_2$ are nonzero with first nonzero value $1$, and
  \item $|A_1 \otimes A_2| = |A_1| \cdot |A_2|$.
\end{enumerate}
This, applied iteratively $m-1$ times, completes the proof of the claim.

Let $A_1$ and $A_2$ be sets as above.
For the first item, consider a vector $v_1 \otimes v_2 \in A_1 \otimes A_2$ with $v_1 \in A_1$ and $v_2 \in A_2$.
Since $v_1$ and $v_2$ are nonzero, it follows that $v_1 \otimes v_2$ is nonzero as well. Additionally, the first nonzero value of $v_1 \otimes v_2$ is the product of the first nonzero values of $v_1$ and $v_2$, and is thus equal to $1$.

For the second item, it suffices to show that for all vectors $u_1, v_1 \in A_1$ and $u_2, v_2 \in A_2$, if $u_1 \otimes u_2 = v_1 \otimes v_2$ then $u_1 = v_1$ and $u_2 = v_2$.
So suppose that $u_1 \otimes u_2 = v_1 \otimes v_2$.
For contradiction, suppose further that $u_1 \neq v_1$, and consider the blocks of length $t_2$ in $u_1 \otimes u_2$ and $v_1 \otimes v_2$ that correspond to the first entry in which $u_1$ and $v_1$ differ. These blocks must be distinct, because $u_2$ and $v_2$ are nonzero vectors whose first nonzero value is $1$, which implies that one is not a multiple of the other by a field element different from $1$. This clearly contradicts the assumption $u_1 \otimes u_2 = v_1 \otimes v_2$.
Now, given that $u_1=v_1$, consider the blocks of length $t_2$ in $u_1 \otimes u_2$ and $v_1 \otimes v_2$ that correspond to the first nonzero value of $u_1$ (and $v_1$). Since those blocks are equal to $u_2$ and $v_2$ respectively, the assumption $u_1 \otimes u_2 = v_1 \otimes v_2$ implies that $u_2=v_2$, and we are done.
\end{proof}

\section{Nearly Orthogonal Sets over the Binary Field}\label{sec:F2}

In this section, we prove the existence of large $k$-nearly orthogonal sets over the binary field and confirm Theorem~\ref{thm:F2}.
We start by proving that for every integer $t$, there exists a relatively small collection of relatively small subsets of $\Fset_2^t$, such that every set of non-self-orthogonal vectors of $\Fset_2^t$ that are pairwise non-orthogonal is contained in one of them.

\begin{lemma}\label{lemma:collectionC_F2}
Let $t$ be an integer.
There exists a collection $\calC$ of subsets of $\Fset_2^t$ such that
\begin{enumerate}
  \item\label{itm:size2} $|\calC| \leq 2^{t^2}$,
  \item\label{itm:size_pairs2} for every $C \in \calC$, it holds that $|C| \leq 2^{(t+1)/2}$, and
  \item\label{itm:contain2} for every set $A \subseteq \Fset_2^t$ with $\langle v_1,v_2 \rangle \neq 0$ for all $v_1,v_2 \in A$, there exists a set $C \in \calC$ such that $A \subseteq C$.
\end{enumerate}
\end{lemma}

\begin{proof}
Fix an integer $t$, and define $\calC$ as the collection of all subspaces of $\Fset_2^t$ of dimension $\lfloor \tfrac{t+1}{2} \rfloor$. We prove that $\calC$ satisfies the three properties required by the lemma.
Firstly, the number of subspaces of $\Fset_2^t$ clearly does not exceed the number of possible bases of such subspaces, hence it holds that $|\calC| \leq 2^{t^2}$, as required for Item~\ref{itm:size2} of the lemma.
Secondly, every subspace of $\Fset_2^t$ of dimension $\lfloor \tfrac{t+1}{2} \rfloor$ includes $2^{\lfloor (t+1)/2 \rfloor}$ vectors, implying that $\calC$ satisfies Item~\ref{itm:size_pairs2} of the lemma.

Finally, to prove that $\calC$ satisfies Item~\ref{itm:contain2} of the lemma, let $A$ be a subset of $\Fset_2^t$ with $\langle v_1,v_2 \rangle \neq 0$ for all $v_1,v_2 \in A$.
Consider the subset $A' = \{(v,1) \mid v \in A\}$ of $\Fset_2^{t+1}$ obtained by adding a $1$ entry at the end of each vector of $A$.
Observe that our assumption on $A$ implies that for every two (not necessarily distinct) vectors $u_1=(v_1,1)$ and $u_2=(v_2,1)$ of $A'$ with $v_1,v_2 \in A$, it holds that $\langle u_1,u_2 \rangle = \langle v_1, v_2 \rangle +1 = 0$.
Let $W = \linspan(A')$ denote the subspace of $\Fset_2^{t+1}$ spanned by the vectors of $A'$.
Since every two vectors of $A'$ are orthogonal, it follows that every two vectors of $W$ are orthogonal. Letting $W^\perp$ stand for the orthogonal complement of $W$, this implies that $W \subseteq W^\perp$, and thus $\dim(W) \leq \dim(W^\perp)$. By $\dim(W)+\dim(W^\perp)=t+1$, we derive that $\dim(W) \leq \lfloor \tfrac{t+1}{2} \rfloor$.
Now, let $\widetilde{W}$ denote the projection of $W$ on the first $t$ coordinates (omitting the last one). Observe that $\widetilde{W}$ forms a subspace of $\Fset_2^t$ of dimension at most $\lfloor \tfrac{t+1}{2} \rfloor$ and that $\widetilde{W} = \linspan(A)$, so in particular, $A \subseteq \widetilde{W}$. By the definition of $\calC$, this yields that there exists a set $C \in \calC$ such that $A \subseteq C$. This completes the proof.
\end{proof}

We are ready to prove Theorem~\ref{thm:F2}.

\begin{proof}[ of Theorem~\ref{thm:F2}]
For an integer $t$, let $V$ denote the set of all non-self-orthogonal vectors of $\Fset_2^t$, that is,
\[V = \{ v \in \Fset_2^t \mid \langle v , v \rangle =1\},\]
and notice that $|V| = 2^{t-1}$.
For an integer $m$, let $Q$ denote the set of all vectors obtained by applying the tensor product operation on $m$ vectors of $V$, that is, $Q = V^{\otimes m} \subseteq \Fset_2^{t^m}$.
Since the vectors of $V$ are nonzero, we can apply Claim~\ref{claim:|V|^m} to obtain that
\begin{eqnarray}\label{eq:|Q|_F2}
|Q| = |V|^m = 2^{m \cdot (t-1)}.
\end{eqnarray}
Note that the vectors of $Q$ are non-self-orthogonal, because every vector $v = v_1 \otimes \cdots \otimes v_m \in Q$ satisfies $\langle v, v \rangle = \prod_{j=1}^{m}{\langle v_j,v_j \rangle} = 1$, where the last equality holds because the vectors $v_1, \ldots, v_m$ are members of $V$.

For an integer $n$, let $\calZ = (z_1, \ldots, z_n)$ be a random sequence of $n$ vectors chosen uniformly and independently from $Q$ (repetitions allowed).
The vectors of $\calZ$ are clearly non-self-orthogonal, because the vectors of $Q$ are.
We will show that for a given integer $k$ and for an appropriate choice of the integers $t$, $m$, and $n$, it holds with positive probability that for every set $I \subseteq [n]$ of size $|I|=k+1$ the vectors of $\{z_i \mid i \in I\}$ include an orthogonal pair.
We start with some preparations.

By Lemma~\ref{lemma:collectionC_F2}, there exists a collection $\calC$ of subsets of $\Fset_2^t$ such that
\begin{enumerate}
  \item\label{itm:size2_T} $|\calC| \leq 2^{t^2}$,
  \item\label{itm:size_pairs2_T} for every $C \in \calC$, it holds that $|C| \leq 2^{(t+1)/2}$, and
  \item\label{itm:contain2_T} for every set $A \subseteq \Fset_2^t$ with $\langle v_1,v_2 \rangle \neq 0$ for all $v_1,v_2 \in A$, there exists a set $C \in \calC$ such that $A \subseteq C$.
\end{enumerate}
Consider the collection
\[ \calB = \{ C^{(1)} \otimes C^{(2)} \otimes \cdots \otimes C^{(m)} \mid C^{(j)} \in \calC \mbox{~for all~}j \in [m] \}.\]
It follows from Item~\ref{itm:size2_T} above that
\begin{eqnarray}\label{eq:|calB|_T_F2}
|\calB| \leq |\calC|^m \leq 2^{m \cdot t^2}.
\end{eqnarray}
It further follows from Item~\ref{itm:size_pairs2_T} that for every set $B = C^{(1)} \otimes C^{(2)} \otimes \cdots \otimes C^{(m)} \in \calB$, it holds that
\begin{eqnarray}\label{eq:|B|_F2}
|B| \leq \prod_{j=1}^{m}{|C^{(j)}|} \leq 2^{m \cdot (t+1)/2}.
\end{eqnarray}
We next claim that for every set $G \subseteq Q$ of pairwise non-orthogonal vectors, there exists a set $B \in \calB$ such that $G \subseteq B$.
To see this, consider such a set $G \subseteq Q$, and for each $j \in [m]$, let $A^{(j)} \subseteq \Fset_2^t$ denote the $j$th projection of $G$ (see Section~\ref{sec:preliminaries}).
Using the property of tensor product given in~\eqref{eq:tensor}, the fact that every two (not necessarily distinct) vectors of $G$ are not orthogonal implies that for each $j \in [m]$, every two (not necessarily distinct) vectors of $A^{(j)}$ are not orthogonal, hence by Item~\ref{itm:contain2_T}, there exists a set $C^{(j)} \in \calC$ such that $A^{(j)} \subseteq C^{(j)}$. This implies that $G \subseteq B$ for the set $B = C^{(1)} \otimes C^{(2)} \otimes \cdots \otimes C^{(m)} \in \calB$, as desired.

Now, for a given integer $k$, consider the event $\calE$ that there exists a set $I \subseteq [n]$ of size $|I|=k+1$ for which the vectors of $\{z_i \mid i \in I\}$ are pairwise non-orthogonal.
As shown above, such vectors lie in some set of the collection $\calB$.
For every fixed $B \in \calB$, we apply the union bound to obtain that the probability that there exists a set $I \subseteq [n]$ of size $|I|=k+1$ such that $\{z_i \mid i \in I\} \subseteq B$ is at most
\[ \binom{n}{k+1} \cdot \bigg ( \frac{|B|}{|Q|} \bigg )^{k+1} \leq \bigg ( \frac{n \cdot |B|}{|Q|} \bigg )^{k+1} \leq \bigg ( \frac{n \cdot 2^{m \cdot (t+1)/2}}{2^{m \cdot (t-1)}} \bigg )^{k+1} = \bigg ( \frac{n}{2^{m \cdot (t-3)/2}} \bigg )^{k+1},\]
where for the second inequality we have used~\eqref{eq:|Q|_F2} and~\eqref{eq:|B|_F2}.
We apply again the union bound, this time over all the sets of $\calB$, and use~\eqref{eq:|calB|_T_F2} to obtain that
\begin{eqnarray}\label{eq:ProbE_F2}
\Prob{}{\calE} \leq |\calB| \cdot \bigg ( \frac{n}{2^{m \cdot (t-3)/2}} \bigg )^{k+1} \leq 2^{m \cdot t^2} \cdot \bigg ( \frac{n}{2^{m \cdot (t-3)/2}} \bigg )^{k+1}.
\end{eqnarray}

We finally set the parameters of the construction, ensuring that the event $\calE$ occurs with probability smaller than $1$.
Let $d \geq k$ be two integers. Note that constant values of $k$ can be handled, using $\Alpha{d}{k}{{\Fset_2}} \geq d$, by an appropriate choice of the constant $\delta$ from the assertion of the theorem.
For a sufficiently large $k$, set $t = \lfloor k/8 \rfloor$, and let $m$ be the largest integer such that $d \geq t^m$. The assumption $d \geq k$ implies that $m \geq 1$. Set $n = \lfloor 2^{m \cdot t/4} \rfloor$.
We obtain that
\[ \Prob{}{\calE} \leq 2^{m \cdot t^2} \cdot \bigg ( \frac{1}{2^{m \cdot (t/4-3/2)}} \bigg )^{k+1} \leq 2^{m \cdot t^2} \cdot \bigg ( \frac{1}{2^{m \cdot t/8}} \bigg )^{k+1}< 1,\]
where the first inequality holds by combining~\eqref{eq:ProbE_F2} with our choice of $n$, the second by the assumption that $k$ and $t$ are sufficiently large (specifically, $t/4-3/2 \geq t/8$ for $t \geq 12$), and the third by our choice of $t$.
This implies that there exists a choice for the $n$ vectors of the sequence $\calZ$ for which the event $\calE$ does not occur, that is, no $k+1$ of them are pairwise non-orthogonal.
For this choice, let $\calG = \{ z_i \mid i \in [n] \} \subseteq \Fset_2^{t^m}$ denote the set that consists of the vectors of $\calZ$. Since the event $\calE$ does not occur, no vector appears in $\calZ$ more than $k$ times, hence $|\calG| \geq n/k$.
We thus obtain, using the monotonicity of $\Alpha{d}{k}{{\Fset_2}}$ with respect to $d$, that
\[ \Alpha{d}{k}{{\Fset_2}} \geq \Alpha{t^m}{k}{{\Fset_2}} \geq n/k \geq 2^{\Omega(m \cdot t)} \geq 2^{\Omega( (\log d) \cdot t/\log t)} \geq d^{\Omega(t/\log t)} \geq d^{\Omega(k/ \log k)}.\]
This completes the proof.
\end{proof}

\section{Nearly Orthogonal Sets over General Finite Fields}\label{sec:F}

In this section, we prove the existence of large $k$-nearly orthogonal sets over general finite fields and confirm Theorem~\ref{thm:GeneralFields}.
As mentioned earlier, we establish a bipartite analogue of the theorem, stated as follows.
\begin{theorem}\label{thm:GeneralFieldsBipartite}
For every prime $p$ there exists a constant $\delta = \delta(p) >0$, such that for every field $\Fset$ of characteristic $p$ and for all integers $k \geq 2$ and $d \geq k^{1/(p-1)}$, the following holds.
There exists a set $\calG$ of non-self-orthogonal vectors of $\Fset^d$ of size
\[ |\calG| \geq d^{\delta \cdot k^{1/(p-1)}/\log k},\]
such that for every two sets $G_1, G_2 \subseteq \calG$ with $|G_1|=|G_2|=k$, there exist vectors $v_1 \in G_1$ and $v_2 \in G_2$ with $\langle v_1,v_2 \rangle = 0$.
\end{theorem}
\noindent
Observe that the set $\calG$ provided by Theorem~\ref{thm:GeneralFieldsBipartite} is $(k-1)$-nearly orthogonal.
Therefore, Theorem~\ref{thm:GeneralFields} can be derived from Theorem~\ref{thm:GeneralFieldsBipartite}.

\subsection{Pseudo-random Graphs}\label{sec:spectral}

An $(n,d,\lambda)$-graph is a $d$-regular graph on $n$ vertices, such that the absolute values of all the eigenvalues of its adjacency matrix but the largest one are at most $\lambda$.
The graphs considered here may have loops, at most one at each vertex, contributing $1$ to the degree of the corresponding vertex.
It is well known that $(n,d,\lambda)$-graphs with $\lambda$ significantly smaller than $d$ satisfy various pseudo-random properties.
One property is given by the following theorem, which says that in such graphs, the number of edges connecting two sets of vertices is close to the expected number of edges between them in a random graph with edge probability $d/n$.

\begin{theorem}\label{thm:pseudo_B,C}
Let $G$ be an $(n,d,\lambda)$-graph (with loops allowed). Then, for every two sets of vertices $C_1$ and $C_2$ of $G$, the number $e(C_1,C_2)$ of pairs $(x_1,x_2)$ of adjacent vertices with $x_1 \in C_1$ and $x_2 \in C_2$ satisfies
\[ \Big | e(C_1,C_2) - \frac{d}{n} \cdot |C_1| \cdot |C_2| \Big | \leq \lambda \cdot \sqrt{|C_1| \cdot |C_2|}.\]
\end{theorem}
\noindent
The proof of Theorem~\ref{thm:pseudo_B,C} is presented in~\cite[Chapter~9.2]{AlonS16} for simple graphs.
The same proof extends to the case where at most one loop is allowed at every vertex.
For completeness, we describe the proof in Appendix~\ref{app:spectral} and verify the extension stated above.

For a prime $p$ and an integer $t$, let $G(p,t)$ denote the graph whose vertices are all the nonzero vectors of $\Fset_p^t$, where two such (not necessarily distinct) vectors $v_1,v_2 \in \Fset_p^t$ are connected by an edge if and only if they are orthogonal, that is, $\langle v_1,v_2 \rangle = 0$.
The following theorem follows from a result of~\cite{Vinh08a} (see also~\cite{AlonK97}).

\begin{theorem}[\cite{Vinh08a}]\label{thm:Vinh}
For every prime $p$ and for every integer $t$, the graph $G(p,t)$ is a
\[(p^t-1,p^{t-1}-1,(p-1) \cdot p^{t/2-1}) \mbox{-graph}.\]
\end{theorem}

By combining Theorems~\ref{thm:pseudo_B,C} and~\ref{thm:Vinh}, we obtain the following result.
\begin{theorem}\label{thm:vectors_B,C}
For a prime $p$ and an integer $t \geq 2$, let $C_1,C_2 \subseteq \Fset_p^t$ be two sets of vectors such that $\langle v_1,v_2 \rangle \neq 0$ for all $v_1 \in C_1$ and $v_2 \in C_2$. Then, $|C_1| \cdot |C_2| \leq p^{t+2}$.
\end{theorem}

\begin{proof}
By Theorem~\ref{thm:Vinh}, the graph $G(p,t)$ is an $(n,d,\lambda)$-graph for $n = p^t-1$, $d = p^{t-1}-1$, and $\lambda = (p-1) \cdot p^{t/2-1}$.
Let $C_1,C_2 \subseteq \Fset_p^t$ be two sets of vectors such that $\langle v_1,v_2 \rangle \neq 0$ for all $v_1 \in C_1$ and $v_2 \in C_2$.
If either $C_1$ or $C_2$ is empty, then the assertion of the theorem trivially holds. Otherwise, the vectors of $C_1$ and $C_2$ are nonzero and thus form vertices of $G(p,t)$.
In this case, our assumption on $C_1$ and $C_2$ implies that the number $e(C_1,C_2)$ of pairs $(x_1,x_2)$ of adjacent vertices in $G(p,t)$ with $x_1 \in C_1$ and $x_2 \in C_2$ is $0$.
Using the assumption $t \geq 2$, we apply Theorem~\ref{thm:pseudo_B,C} to obtain that
\[ |C_1| \cdot |C_2| \leq \bigg ( \frac{n \cdot \lambda}{d} \bigg)^2 = \bigg ( \frac{(p^t-1) \cdot (p-1) \cdot p^{t/2-1}}{p^{t-1}-1} \bigg )^2 \leq ((p^2-1) \cdot p^{t/2-1})^2 \leq p^{t+2}.\]
This completes the proof.
\end{proof}

\subsection{A Key Lemma}

Equipped with Theorem~\ref{thm:vectors_B,C}, we are ready to prove the following key lemma.

\begin{lemma}\label{lemma:collectionC}
Let $p$ be a prime, and let $t \geq 2$ be an integer.
There exists a collection $\calC$ of pairs of subsets of $\Fset_p^t$ such that
\begin{enumerate}
  \item\label{itm:size} $|\calC| \leq p^{2t \cdot (t^{p-1}+p-1)}$,
  \item\label{itm:size_pairs} for every pair $(C_1,C_2) \in \calC$, it holds that $|C_1| \cdot |C_2| \leq p^{t+2}$, and
  \item\label{itm:contain} for every pair $(A_1, A_2)$ of subsets of $\Fset_p^t$ with $\langle v_1,v_2 \rangle \neq 0$ for all $v_1 \in A_1$ and $v_2 \in A_2$, there exists a pair $(C_1,C_2) \in \calC$ such that $A_1 \subseteq C_1$ and $A_2 \subseteq C_2$.
\end{enumerate}
\end{lemma}

\begin{proof}
Fix a prime $p$ and an integer $t \geq 2$.
We start with some definitions.
Let $g: \Fset_p^t \rightarrow \Fset_p^{t^{p-1}+p-1}$ denote the function that maps every vector $v \in \Fset_p^t$ to the vector
\[g(v) = (v^{\otimes p-1}, 1, \ldots, 1),\]
defined as the tensor product of $v$ with itself $p-1$ times followed by $p-1$ ones.
Observe that for every two vectors $v_1,v_2 \in \Fset_p^t$, it holds that
\begin{eqnarray}\label{eq:<g,g>}
\langle g(v_1), g(v_2) \rangle = \langle v_1^{\otimes p-1},v_2^{\otimes p-1} \rangle +(p-1) = \langle v_1, v_2 \rangle^{p-1} + (p-1).
\end{eqnarray}
For any set $W \subseteq \Fset^{t^{p-1}+p-1}$, we let
\[g^{-1}(W) = \{ v \in \Fset_p^t \mid g(v) \in W \}.\]
We define $\calC$ as the collection of all pairs $(C_1,C_2)$ defined by $C_1 = g^{-1}(W_1)$ and $C_2 = g^{-1}(W_2)$, where $W_1$ and $W_2$ are some orthogonal subspaces of $\Fset_p^{t^{p-1}+p-1}$ that are spanned by vectors of the image of the function $g$. We turn to proving that $\calC$ satisfies the three properties required by the lemma.

First, observe that the number of subspaces of $\Fset_p^{t^{p-1}+p-1}$ that are spanned by vectors of the image of the function $g$ does not exceed the number of possible choices of $t^{p-1}+p-1$ vectors from the image of $g$. Since the size of this image is at most $p^t$, the number of those subspaces is bounded by $(p^t)^{t^{p-1}+p-1} = p^{t \cdot (t^{p-1}+p-1)}$. The size of $\calC$ is bounded by the number of pairs of such subspaces, hence $|\calC| \leq p^{2t \cdot (t^{p-1}+p-1)}$, as required for Item~\ref{itm:size} of the lemma.

We proceed by proving that for every pair $(C_1,C_2) \in \calC$, it holds that $\langle v_1, v_2 \rangle \neq 0$ for all $v_1 \in C_1$ and $v_2 \in C_2$.
Consider a pair $(C_1,C_2) \in \calC$.
By definition, there exist orthogonal subspaces $W_1, W_2 \subseteq \Fset_p^{t^{p-1}+p-1}$ such that $C_1 = g^{-1}(W_1)$ and $C_2 = g^{-1}(W_2)$.
Consider two vectors $v_1 \in C_1$ and $v_2 \in C_2$. By definition, $g(v_1) \in W_1$ and $g(v_2) \in W_2$. Since $W_1$ and $W_2$ are orthogonal, it follows that $\langle g(v_1), g(v_2) \rangle = 0$. Using~\eqref{eq:<g,g>}, it follows that $\langle v_1, v_2 \rangle^{p-1} = 1$, hence $\langle v_1, v_2 \rangle \neq 0$, as required.
This allows us to apply Theorem~\ref{thm:vectors_B,C} and to obtain that $|C_1| \cdot |C_2| \leq p^{t+2}$, as required for Item~\ref{itm:size_pairs} of the lemma.

We finally prove that $\calC$ satisfies Item~\ref{itm:contain} of the lemma.
Let $(A_1, A_2)$ be a pair of subsets of $\Fset_p^t$ with $\langle v_1,v_2 \rangle \neq 0$ for all $v_1 \in A_1$ and $v_2 \in A_2$.
Consider the subsets $A'_1 = \{ g(v) \mid v \in A_1 \}$ and $A'_2 = \{ g(v) \mid v \in A_2 \}$ of $\Fset_p^{t^{p-1}+p-1}$.
Let $W_1 = \linspan(A'_1)$ and $W_2 = \linspan(A'_2)$ be the subspaces spanned by the vectors of $A'_1$ and $A'_2$ respectively, and notice that they are spanned by vectors of the image of $g$.
It clearly holds that $A_1 \subseteq g^{-1}(W_1)$ and $A_2 \subseteq g^{-1}(W_2)$.
We claim that the subspaces $W_1$ and $W_2$ are orthogonal. To this end, it suffices to show that the vectors of $A'_1$ are orthogonal to those of $A'_2$.
Consider two vectors $u_1 \in A'_1$ and $u_2 \in A'_2$. By definition, they satisfy $u_1 = g(v_1)$ and $u_2 = g(v_2)$ for some $v_1 \in A_1$ and $v_2 \in A_2$.
By our assumption on the pair $(A_1,A_2)$, it holds that $\langle v_1,v_2 \rangle \neq 0$, which implies using Fermat's little theorem that $\langle v_1, v_2 \rangle^{p-1} =1$.
Using~\eqref{eq:<g,g>}, it follows that
\[\langle u_1, u_2 \rangle = \langle g(v_1), g(v_2) \rangle = \langle v_1, v_2 \rangle^{p-1} + (p-1) = 0,\]
as desired.
Since $W_1$ and $W_2$ are orthogonal subspaces of $\Fset_p^{t^{p-1}+p-1}$ that are spanned by vectors of the image of the function $g$, there exists a pair $(C_1,C_2) \in \calC$ with $C_1 = g^{-1}(W_1)$ and $C_2 = g^{-1}(W_2)$.
This pair satisfies that $A_1 \subseteq C_1$ and $A_2 \subseteq C_2$, so we are done.
\end{proof}

\begin{remark}
It is well known that every element of a finite field is expressible as a sum of two squares.
Therefore, one can define the function $g$ in the above proof with the $p-1$ ones replaced by two fixed field elements whose sum of squares is $p-1$.
This change slightly decreases the size of the collection $\calC$ for $p \geq 5$, but makes no difference for our applications of Lemma~\ref{lemma:collectionC}.
\end{remark}

As a simple application of Lemma~\ref{lemma:collectionC}, we determine the number of sets of non-self-orthogonal vectors that are pairwise non-orthogonal in vector spaces over prime order fields. Note that this result is not needed for the subsequent proofs.

\begin{theorem}\label{thm:counting}
Let $p$ be a fixed prime.
Then, for all integers $t$, the number of sets of non-self-orthogonal vectors of $\Fset_p^t$ that are pairwise non-orthogonal is $2^{\Theta( p^{t/2})}$.
\end{theorem}

\begin{proof}
Fix a constant prime $p$ and an integer $t \geq 2$, and let $N_{p,t}$ denote the number of sets of non-self-orthogonal vectors of $\Fset_p^t$ that are pairwise non-orthogonal.
We start with a lower bound on $N_{p,t}$.
It is shown in~\cite{Zame70} that there exists a set $A \subseteq \Fset_p^{t-1}$ of size $|A| \geq p^{t/2 - 2}$ such that $\langle x,y \rangle =0$ for all $x,y \in A$. Let $B \subseteq \Fset_p^{t}$ denote the set of vectors obtained by adding a $1$ entry at the end of the vectors of $A$, and observe that $\langle x,y \rangle =1$ for all $x,y \in B$.
Since the vectors of every subset of $B$ are non-self-orthogonal and pairwise non-orthogonal, it follows that $N_{p,t} \geq 2^{|B|} = 2^{|A|} \geq 2^{\Omega( p^{t/2})}$.

We proceed with an upper bound on $N_{p,t}$. Lemma~\ref{lemma:collectionC} implies that there exists a collection $\calC$ of pairs of subsets of $\Fset_p^t$ such that
\begin{itemize}
  \item $|\calC| \leq p^{O(t^{p})}$,
  \item for every pair $(C_1,C_2) \in \calC$, it holds that $\min(|C_1|,|C_2|) \leq p^{t/2+1}$, and
  \item for every set $A$ of non-self-orthogonal vectors of $\Fset_p^t$ that are pairwise non-orthogonal, there exists a pair $(C_1,C_2) \in \calC$ such that $A \subseteq C_1$ and $A \subseteq C_2$.
\end{itemize}
It follows that every set of non-self-orthogonal vectors of $\Fset_p^t$ that are pairwise non-orthogonal is a subset of some set of size at most $p^{t/2+1}$ that appears in the pairs of $\calC$. We therefore obtain that
\[ N_{p,t} \leq 2 \cdot |\calC| \cdot 2^{p^{t/2+1}} \leq p^{O(t^{p})} \cdot 2^{p^{t/2+1}} \leq 2^{O( p^{t/2})}.\]
This completes the proof.
\end{proof}

\subsection{Proof of Theorem~\ref{thm:GeneralFieldsBipartite}}

With Lemma~\ref{lemma:collectionC} at hand, we turn to the proof of Theorem~\ref{thm:GeneralFieldsBipartite}.

\begin{proof}[ of Theorem~\ref{thm:GeneralFieldsBipartite}]
It suffices to prove the theorem for prime order fields, because every field of characteristic $p$ contains the field of order $p$ as a sub-field.
Let $p$ be a prime, and consider the field $\Fset_p$ of $p$ elements.
For an integer $t \geq 2$, let $V$ denote the set of all non-self-orthogonal vectors of $\Fset_p^t$ whose first nonzero value is $1$.
We observe that
\begin{eqnarray}\label{eq:|V|}
|V| \geq \frac{p^{t-1}-1}{p-1} \geq p^{t-2}.
\end{eqnarray}
Indeed, for every choice of the first $t-1$ values of a vector of $\Fset_p^t$, it is possible to choose a value for the remaining entry (say, $0$ or $1$) to obtain a non-self-orthogonal vector. The bound in~\eqref{eq:|V|} follows by observing that the number of nonzero vectors of $\Fset_p^{t-1}$ in which the first nonzero value is $1$ is $(p^{t-1}-1)/(p-1)$.

For an integer $m$, let $Q$ denote the set of all vectors obtained by applying the tensor product operation on $m$ vectors of $V$, that is,
$Q = V^{\otimes m} \subseteq \Fset_p^{t^m}$.
Since the vectors of $V$ are nonzero and have $1$ as their first nonzero value, we can apply Claim~\ref{claim:|V|^m} to obtain that $|Q| = |V|^m$. Combined with~\eqref{eq:|V|}, this yields that
\begin{eqnarray}\label{eq:|Q|}
|Q| \geq p^{m \cdot (t-2)}.
\end{eqnarray}
Note that the vectors of $Q$ are non-self-orthogonal, because every vector $v = v_1 \otimes \cdots \otimes v_m \in Q$ satisfies $\langle v, v \rangle = \prod_{j=1}^{m}{\langle v_j,v_j \rangle} \neq 0$, where the inequality holds because the vectors $v_1, \ldots, v_m$ are members of $V$.

For an integer $n$, let $\calZ = (z_1, \ldots, z_n)$ be a random sequence of $n$ vectors chosen uniformly and independently from $Q$ (repetitions allowed).
The vectors of $\calZ$ are clearly non-self-orthogonal, because the vectors of $Q$ are.
We will show that for a given integer $k$ and for an appropriate choice of the integers $t$, $m$, and $n$, the set that consists of the vectors of $\calZ$ satisfies with positive probability the property declared in the theorem.
We start with some preparations.

By Lemma~\ref{lemma:collectionC}, there exists a collection $\calC$ of pairs of subsets of $\Fset_p^t$ such that
\begin{enumerate}
  \item\label{itm:sizeT} $|\calC| \leq p^{2t \cdot (t^{p-1}+p-1)}$,
  \item\label{itm:size_pairsT} for every pair $(C_1,C_2) \in \calC$, it holds that $|C_1| \cdot |C_2| \leq p^{t+2}$, and
  \item\label{itm:containT} for every pair $(A_1, A_2)$ of subsets of $\Fset_p^t$ with $\langle v_1,v_2 \rangle \neq 0$ for all $v_1 \in A_1$ and $v_2 \in A_2$, there exists a pair $(C_1,C_2) \in \calC$ such that $A_1 \subseteq C_1$ and $A_2 \subseteq C_2$.
\end{enumerate}
Consider the collection $\calB$ of all pairs $(B_1,B_2)$ of subsets of $\Fset_p^{t^m}$ of the form
\begin{eqnarray}\label{eq:B1B2}
B_1 = C_1^{(1)} \otimes C_1^{(2)} \otimes \cdots \otimes C_1^{(m)} \mbox{~~and~~} B_2 = C_2^{(1)} \otimes C_2^{(2)} \otimes \cdots \otimes C_2^{(m)},
\end{eqnarray}
where $(C^{(1)}_1,C^{(1)}_2), (C^{(2)}_1,C^{(2)}_2), \ldots, (C^{(m)}_1,C^{(m)}_2)$ are $m$ pairs of the collection $\calC$.
It follows from Item~\ref{itm:sizeT} above that
\begin{eqnarray}\label{eq:|calB|}
|\calB| \leq |\calC|^m \leq p^{2mt \cdot (t^{p-1}+p-1)}.
\end{eqnarray}
It further follows from Item~\ref{itm:size_pairsT} that for every pair $(B_1,B_2) \in \calB$ as in~\eqref{eq:B1B2}, it holds that
\begin{eqnarray}\label{eq:|B1||B2|}
|B_1| \cdot |B_2| \leq \prod_{j=1}^{m}{|C_1^{(j)}|} \cdot \prod_{j=1}^{m}{|C_2^{(j)}|} = \prod_{j=1}^{m}{|C_1^{(j)}| \cdot |C_2^{(j)}|} \leq p^{m \cdot (t+2)}.
\end{eqnarray}
We next claim that for every pair $(G_1,G_2)$ of subsets of $Q$ with $\langle v_1,v_2 \rangle \neq 0$ for all $v_1 \in G_1$ and $v_2 \in G_2$, there exists a pair $(B_1,B_2) \in \calB$ such that $G_1 \subseteq B_1$ and $G_2 \subseteq B_2$.
To see this, consider such a pair $(G_1,G_2)$ of subsets of $Q$.
For each $j \in [m]$, let $A^{(j)}_1 \subseteq \Fset_p^t$ and $A^{(j)}_2 \subseteq \Fset_p^t$ denote the $j$th projections of $G_1$ and $G_2$ respectively (see Section~\ref{sec:preliminaries}).
Using the property of tensor product given in~\eqref{eq:tensor}, the fact that the vectors of $G_1$ are not orthogonal to those of $G_2$ implies that for each $j \in [m]$, the vectors of $A^{(j)}_1$ are not orthogonal to those of $A^{(j)}_2$, hence by Item~\ref{itm:containT}, there exists a pair $(C^{(j)}_1,C^{(j)}_2) \in \calC$ such that $A^{(j)}_1 \subseteq C^{(j)}_1$ and $A^{(j)}_2 \subseteq C^{(j)}_2$. This implies that for some pair $(B_1,B_2) \in \calB$, defined as in~\eqref{eq:B1B2}, it holds that $G_1 \subseteq B_1$ and $G_2 \subseteq B_2$, as desired.

For a given integer $k$, consider the event $\calE$ that there exists a pair $(I_1,I_2)$ of subsets of $[n]$ with $|I_1|=|I_2|=k$, such that $\langle z_{i_1},z_{i_2} \rangle \neq 0$ for all $i_1 \in I_1$ and $i_2 \in I_2$.
The above discussion implies that for every such pair $(I_1,I_2)$, there exists a pair $(B_1,B_2) \in \calB$ such that $\{z_i \mid i \in I_1\} \subseteq B_1$ and $\{z_i \mid i \in I_2\} \subseteq B_2$.
For a fixed pair $(B_1,B_2) \in \calB$, suppose without loss of generality that $|B_1| \leq |B_2|$, and apply the union bound to obtain that the probability that there exists a set $I \subseteq [n]$ with $|I|=k$ such that $\{z_i \mid i \in I\} \subseteq B_1$ is at most
\[ \binom{n}{k} \cdot \bigg ( \frac{|B_1|}{|Q|} \bigg )^k \leq \bigg ( \frac{n \cdot |B_1| }{|Q|} \bigg )^k
\leq \bigg ( \frac{n \cdot p^{m \cdot (t+2)/2}}{p^{m \cdot (t-2)}} \bigg )^k = \bigg ( \frac{n}{p^{m \cdot (t/2-3)}} \bigg )^k, \]
where for the second inequality we have used~\eqref{eq:|Q|} and~\eqref{eq:|B1||B2|}.
We apply again the union bound, this time over all pairs of $\calB$, and use~\eqref{eq:|calB|} to obtain that
\begin{eqnarray}\label{eq:Prob[E]p}
\Prob{}{\calE} \leq |\calB| \cdot \bigg ( \frac{n}{p^{m \cdot (t/2-3)}} \bigg )^k \leq p^{2mt \cdot (t^{p-1}+p-1)} \cdot \bigg ( \frac{n}{p^{m \cdot (t/2-3)}} \bigg )^k.
\end{eqnarray}

We finally set the parameters of the construction, ensuring that the event $\calE$ occurs with probability smaller than $1$.
Let $d$ and $k$ be two integers satisfying $d \geq k^{1/(p-1)}$. Note that constant values of $k$ can be handled, using the $d$ vectors of the standard basis of $\Fset_p^d$, by an appropriate choice of the constant $\delta$ from the assertion of the theorem.
Let $t$ be the largest integer satisfying $k > 32 \cdot t^{p-1}$, and let $m$ be the largest integer satisfying $d \geq t^m$.
Assuming that $k$ is sufficiently large, we have $t \geq 2$, and by the assumption $d \geq k^{1/(p-1)}$, we have $m \geq 1$.
Set $n = \lfloor p^{m \cdot t/4} \rfloor$.
We obtain that
\[ \Prob{}{\calE} \leq p^{2mt \cdot (t^{p-1}+p-1)} \cdot \bigg ( \frac{1}{p^{m \cdot (t/4-3)}} \bigg )^k \leq p^{4m \cdot t^{p}} \cdot \bigg ( \frac{1}{p^{m \cdot t/8}} \bigg )^k  <1,\]
where the first inequality holds by combining~\eqref{eq:Prob[E]p} with our choice of $n$, the second by the assumption that $k$ and $t$ are sufficiently large (specifically, $p-1 \leq t^{p-1}$ for $t \geq 2$, and $t/4-3 \geq t/8$ for $t \geq 24$), and the third by our choice of $t$.
This implies that there exists a choice for the $n$ vectors of the sequence $\calZ$ for which the event $\calE$ does not occur, that is, for every pair $(I_1,I_2)$ of subsets of $[n]$ with $|I_1|=|I_2|=k$, there exist indices $i_1 \in I_1$ and $i_2 \in I_2$ with $\langle z_{i_1},z_{i_2} \rangle = 0$.
For this choice, let $\calG = \{z_i \mid i \in [n]\}  \subseteq \Fset_p^{t^m}$ denote the set that consists of the vectors of $\calZ$. It follows that for every two sets $G_1,G_2 \subseteq \calG$ with $|G_1| = |G_2|=k$, there exist vectors $v_1 \in G_1$ and $v_2 \in G_2$ with $\langle v_1, v_2 \rangle = 0$. It further follows that no vector appears in $\calZ$ more than $k-1$ times, hence $|\calG| \geq n/(k-1)$.
We thus obtain that
\[ |\calG| \geq n/(k-1) \geq p^{\Omega(m \cdot t)} \geq p^{\Omega((\log d ) \cdot t/\log t)} \geq d^{\Omega(t/\log t)} \geq d^{\Omega(k^{1/(p-1)}/ \log k)} ,\]
where the $\Omega$ notation hides constants that depend solely on $p$.
By adding $d-t^m$ zero entries at the end of the vectors of $\calG$, we obtain the desired subset of $\Fset_p^d$, and the proof is completed.
\end{proof}

\subsection{Orthogonal Representations}\label{sec:OD}

In this section, we present some applications of our results to the context of orthogonal representations of graphs.
A $d$-dimensional orthogonal representation of a graph $G=(V,E)$ over a field $\Fset$ is an assignment of a non-self-orthogonal vector $u_x \in \Fset^d$ to each vertex $x \in V$, such that $\langle u_x, u_y \rangle = 0$ whenever $x$ and $y$ are distinct non-adjacent vertices of $G$.
For a graph $G$ and a field $\Fset$, let $\xi_\Fset(G)$ denote the smallest integer $d$ for which $G$ admits a $d$-dimensional orthogonal representation over $\Fset$. We start with the following result.

\begin{theorem}\label{thm:K_k,k}
For every prime $p$ there exists a constant $c=c(p) >0$ such that for every field $\Fset$ of characteristic $p$ and for every integer $k \geq 2$, there are infinitely many integers $n$ for which there exists a $K_{k,k}$-free graph $G$ on $n$ vertices such that
\[ \xi_\Fset(G) \leq n^{c \cdot (\log k) / k^{1/(p-1)}}.\]
\end{theorem}

\begin{proof}
Let $p$ be a prime, let $\Fset$ be a field of characteristic $p$, and let $k \geq 2$ be an integer.
By Theorem~\ref{thm:GeneralFieldsBipartite}, there exists some $\delta = \delta(p) >0$ such that for every integer $d \geq k^{1/(p-1)}$, there exists a set $\calG$ of non-self-orthogonal vectors of $\Fset^d$ of size
\[ |\calG| \geq d^{\delta \cdot k^{1/(p-1)}/\log k},\]
such that for every two sets $G_1, G_2 \subseteq \calG$ with $|G_1|=|G_2|=k$, there exist vectors $v_1 \in G_1$ and $v_2 \in G_2$ with $\langle v_1,v_2 \rangle = 0$.
Consider the graph $G$ on the vertex set $\calG$, where two distinct vertices are adjacent if and only if their vectors are not orthogonal, and let $n$ denote the number of its vertices.
It follows that the graph $G$ is $K_{k,k}$-free and that $\xi_\Fset(G) \leq d \leq n^{c \cdot (\log k) / k^{1/(p-1)}}$ for some $c = c(p)>0$.
This completes the proof.
\end{proof}

The clique cover number of a graph $G$, denoted by $\overline{\chi}(G)$, is the minimum number of cliques needed to cover the vertex set of $G$.
The following theorem provides graphs $G$ with a large ratio between $\overline{\chi}(G)$ and $\xi_\Fset(G)$ for fields $\Fset$ with finite characteristic.

\begin{theorem}\label{thm:CHIvsOD}
For every prime $p$ there exists a constant $c=c(p) >0$ such that for every field $\Fset$ of characteristic $p$, there are infinitely many integers $n$ for which there exists a graph $G$ on $n$ vertices such that
\[ \frac{\overline{\chi}(G)}{\xi_\Fset(G)} \geq c \cdot \frac{n}{\log^p n}.\]
\end{theorem}

\begin{proof}
Let $p$ be a prime, and let $\Fset$ be a field of characteristic $p$.
We apply Theorem~\ref{thm:GeneralFields} with integers $d$ and $k$ satisfying $k = d^{p-1}$.
We obtain that for some $\delta = \delta(p) >0$, there exists a $k$-nearly orthogonal set $\calG \subseteq \Fset^d$ of size
\[ |\calG| \geq d^{\delta \cdot k^{1/(p-1)}/ \log k} = 2^{\frac{\delta}{p-1} \cdot  k^{1/(p-1)}}.\]
Consider the graph $G$ on the vertex set $\calG$, where two distinct vertices are adjacent if and only if their vectors are not orthogonal, and let $n$ denote the number of its vertices.
It follows that the graph $G$ is $K_{k+1}$-free for $k \leq O( \log^{p-1} n)$ and that $\xi_\Fset(G) \leq d$ for $d = k^{1/(p-1)} \leq O(\log n)$.
Since $G$ has no clique of size $k+1$, its clique cover number satisfies $\overline{\chi}(G) \geq n/k$.
We derive that for some $c = c(p)>0$, it holds that
\[ \frac{\overline{\chi}(G)}{\xi_\Fset(G)} \geq \frac{n}{k \cdot d} \geq c \cdot \frac{n}{\log^p n}.\]
This completes the proof.
\end{proof}

\begin{remark}\label{remark:CHI}
We note that for every field $\Fset$ and for every graph $G$ on $n$ vertices, the ratio between $\overline{\chi}(G)$ and $\xi_\Fset(G)$ is bounded from above by $O(n/\log^2 n)$. Indeed, a result of Erd{\H{o}}s~\cite{Erdos67} shows such a bound on the ratio between $\overline{\chi}(G)$ and the independence number of $G$, and the latter is bounded from above by $\xi_\Fset(G)$ for every field $\Fset$.
It thus follows that the bound achieved in Theorem~\ref{thm:CHIvsOD} for the binary field is tight up to a multiplicative constant.
\end{remark}

\section*{Acknowledgments}
We are grateful to the anonymous reviewers for their useful suggestions.

\bibliographystyle{abbrv}
\bibliography{nearly}

\appendix

\section{Proof of Theorem~\ref{thm:pseudo_B,C}}\label{app:spectral}

As mentioned before, Theorem~\ref{thm:pseudo_B,C} is proved in~\cite[Chapter~9.2]{AlonS16} for simple graphs.
We present here the proof and verify that it extends to graphs with loops, at most one at each vertex.
Recall that we adopt the convention that a loop contributes $1$ to the degree of the corresponding vertex.

\begin{proof}[ of Theorem~\ref{thm:pseudo_B,C}]
Let $G=(V,E)$ be an $(n,d,\lambda)$-graph with at most one loop at every vertex, and let $A$ denote the adjacency matrix of $G$.
Since $G$ is $d$-regular, where every loop contributes $1$ to the degree of its vertex, each row of $A$ has precisely $d$ ones.
It thus follows that the largest eigenvalue of $A$ is $d$, corresponding to the all-one vector.

Let $C_1,C_2 \subseteq V$ be two sets of vertices, and put $c = |C_1|/n$.
For a vertex $v \in V$, let $N_{C_1}(v)$ denote the set of all vertices of $C_1$ that are adjacent to $v$ in $G$ (including itself, if it lies in $C_1$ and has a loop).
Let $f: V \rightarrow \R$ denote the vector defined by $f(v)=1-c$ for $v \in C_1$ and $f(v)=-c$ for $v \notin C_1$. Notice that $\sum_{v \in V}{f(v)} = c n \cdot (1-c) - (n-c n) \cdot c = 0$, hence $f$ is orthogonal to the eigenvector of the largest eigenvalue of $A$. This implies that
\begin{eqnarray}\label{eq:<Af,Af>}
\langle Af, Af \rangle \leq \lambda^2 \cdot \langle f,f \rangle.
\end{eqnarray}
Observe that
\[\langle f,f \rangle = c n \cdot (1-c)^2 + (n-c n) \cdot c^2 = c n \cdot (1-c).\]
Observe further that
\[\langle Af,Af \rangle = \sum_{v \in V}{( |N_{C_1}(v)| \cdot (1-c)-(d-|N_{C_1}(v)|) \cdot c )^2} = \sum_{v \in V}{(|N_{C_1}(v)|-c d)^2},\]
where we again use the fact that a loop contributes $1$ to the degree of its vertex.
It therefore follows from~\eqref{eq:<Af,Af>} that
\begin{eqnarray}\label{eq:spectral}
\sum_{v \in V}{(|N_{C_1}(v)|-c d)^2} \leq \lambda^2 \cdot c n \cdot (1-c).
\end{eqnarray}
Using the Cauchy--Schwarz inequality, we derive from~\eqref{eq:spectral} that
\begin{eqnarray*}
\bigg  |e(C_1,C_2) - \frac{d}{n} \cdot |C_1| \cdot |C_2| \bigg | & \leq & \sum_{v \in C_2}{ \big| |N_{C_1}(v)|-cd \big |} \\
&\leq & \sqrt{|C_2|} \cdot \bigg ( \sum_{v \in C_2}{ \big ( |N_{C_1}(v)|-cd \big )^2} \bigg )^{1/2} \\
&\leq & \sqrt{|C_2|} \cdot \bigg ( \sum_{v \in V}{ \big ( |N_{C_1}(v)|-cd \big )^2} \bigg )^{1/2} \\
& \leq & \sqrt{|C_2|} \cdot \lambda \cdot \sqrt{c n \cdot (1-c)} \\
&\leq &  \lambda \cdot  \sqrt{|C_1| \cdot |C_2|}.
\end{eqnarray*}
This completes the proof.
\end{proof}

\end{document}